\documentclass[10pt,letterpaper]{article}

\usepackage{palatino}
\usepackage{xspace}
\usepackage[utf8]{inputenc}
\usepackage{amsmath}
\usepackage{amsthm}
\usepackage{amsfonts}
\usepackage{color}
\usepackage{wrapfig}
\usepackage{xspace}
\usepackage{graphicx}
\usepackage{fullpage}
\usepackage{hyperref}
\usepackage{subcaption}
\usepackage[font=bf]{caption}
\usepackage{url}
\usepackage{color}
\usepackage{mathtools}
\newcommand{\defeq}{\vcentcolon=}

\usepackage[compact]{titlesec}
\usepackage{algorithm}
\usepackage{algpseudocode}
\algtext*{EndWhile}% Remove "end while" text
\algtext*{EndIf}% Remove "end if" text
\algtext*{EndFor}% Remove "end if" text

\newtheorem{lemma}{Lemma}[section]

\newtheorem{claim}[lemma]{Claim}

\newtheorem{theorem}[lemma]{Theorem}

\newtheorem{remark}[lemma]{Remark}

\newcounter{packednmbr}

\usepackage[sort,nocompress]{cite}

\title{Approximations of Schatten Norms via Taylor Expansions}
\author{
  Vladimir Braverman%
  \thanks{This material is based upon work supported in part by the National Science Foundation under
Grants No. 1447639, 1650041 and 1652257, Cisco faculty award, and by the ONR Award N00014-18-1-2364.
Email: \texttt{vova@cs.jhu.edu}
  }
  \\
  Johns Hopkins University
}

\begin{document}
\maketitle

\begin{abstract}
In this paper we consider symmetric, positive semidefinite (SPSD) matrix $A$ and present two algorithms for computing the $p$-Schatten norm $\|A\|_p$. The first algorithm works for any SPSD matrix $A$. The second algorithm works for non-singular SPSD matrices and runs in time that depends on $\kappa = {\lambda_1(A)\over \lambda_n(A)}$, where $\lambda_i(A)$ is the $i$-th eigenvalue of $A$. Our methods are simple and easy to implement and can be extended to general matrices.
Our algorithms improve, for a range of parameters, recent results of Musco, Netrapalli, Sidford, Ubaru and Woodruff (ITCS 2018.) and match the running time of the methods by Han, Malioutov, Avron, and Shin (SISC 2017) while avoiding computations of coefficients of Chebyshev polynomials.
%Our algorithms can be implemented in the streaming model, providing multipass streaming algorithms for non-integral Schatten norms.
\end{abstract}

\section{Introduction}

In many applications of data science and machine learning data is represented by large matrices. Fast and accurate analysis of such matrices is a challenging task that is of paramount importance for the aforementioned applications. Randomized numerical algebra (RNA) is an popular area of research that often provides such fast and accurate algorithmic methods for massive matrix computations. Many critical problems in RNA boil down to approximating spectral functions and one of the most fundamental examples of such spectral functions is Schatten $p$ norm. The $p$-th Schatten norm for matrix $A \in R^{n_1\times n_2}$ is defined as the $l_p$ norm of a vector comprised of singular values of matrix $A$, i.e.,
$$
\|A\|_{p} = \left(\sum_{i=1}^{\min{(n_1,n_2)}} |\sigma_i(A)|^{p}\right)^{1/{p}},
$$
where $\sigma_i(A)$ is the $i$-th singular value of $A$.

%\subsection{Streaming Model}
%
%In the streaming model, we recently provided a new result for the integer values of $p$. In this paper we extend these results to arbitrary values of $p>0$.

\subsection{Our Results and Related Work}
In this paper we consider symmetric, positive semidefinite (SPSD) matrix $A$ and present two algorithms for computing $\|A\|_p$. The first algorithm in Section \ref{sec:one} works for any SPSD matrix $A$. The second algorithm in Section \ref{sec:two} works for non-singular SPSD matrices and runs in time that depends on $\kappa = {\lambda_1(A)\over \lambda_n(A)}$, where $\lambda_i(A)$ is the $i$-th eigenvalue of $A$. The table below summarizes our results.
Our methods are simple and easy to implement and can be extended to general matrices. Indeed, to compute $\|B\|_{p}$ for matrix $B$, one can apply our methods to SPSD matrix $B^TB$ and note that $\|B\|_{p} = \|B^TB\|_{p/2}$.
It is well known that for SPSD matrix $A$, $\|A\|_p^p = \mathbf{tr}\ A^p$ and thus, our algorithms provide multiplicative approximations for $\mathbf{tr}\ A^p$.

Musco, Netrapalli, Sidford, Ubaru and Woodruff \cite{musco_et_al:LIPIcs:2018:8339} (see also the full version in \cite{DBLP:journals/corr/MuscoNSUW17}) provided a general approach for approximating spectral functions that works for general matrices.
Table $1$ summarizes some of their results for Schatten norms for general matrices.
%they use $d_s(A)$ to denote the \emph{universal sparsity assumption} defined \cite{musco_et_al:LIPIcs:2018:8339} as ``... the maximum row sparsity $ds(A) \le \xi n \cdot nnz(A)$ for some constant $\xi$.''
Our result in Theorem \ref{th:21} improves the bounds in \cite{musco_et_al:LIPIcs:2018:8339,DBLP:journals/corr/MuscoNSUW17} for a range of parameters, for example, when $p > 1$ and $\epsilon = o(n^{1\over 1-p^2})$ or when $p \ge 1$ and $nnz(A) = o(n^{1+{p -1\over p(p +1)}})$.

In \cite{doi:10.1137/16M1078148}, Han, Malioutov, Avron, and Shin proposed to use general Chebyshev polynomials to approximate a wide class of functions and Schatten norms in particular. The methods in \cite{doi:10.1137/16M1078148} work for invertible matrices and the running time depends on $\kappa$. Table $1$ summarizes their results for Schatten norms. Under the assumptions from \cite{doi:10.1137/16M1078148} the running time in Remark \ref{cor:23} matches the running time of the algorithm in \cite{doi:10.1137/16M1078148} for constant $p$. In addition, our methods do not need to compute or store the coefficients of Chebyshev polynomials.
A straightforward computation of Chebyshev coefficients for Schatten norm \cite{doi:10.1137/16M1078148} requires time $\Omega(z^2)$ where\footnote{In \cite{doi:10.1137/16M1078148} they use $n$ instead of $z$.} $z = {\kappa} \left(p \log \kappa + \log {1\over \epsilon}\right)$.
Thus, for a range of parameters, e.g. when $nnz(A) = o(\kappa)$, a straightforward computation of Chebyshev coefficients may be more expensive than computing the approximation of the Schatten norm in our paper.

Our work has been inspired by the approach of Boutsidis, Drineas, Kambadur, Kontopoulou, and Zouzias \cite{boutsidis2017randomized} that uses Taylor expansion for $\log(1+x)$ to approximate the log determinant.
The coefficients of this expansion have the same sign and thus the Hutchinson estimator \cite{hutchinson} can be applied to each partial sum of this expansion.
This is not the case for the Taylor expansion of $(1-x)^p$. The key idea of our approach is to find $m$ such a constant minus the partial sum of the first $m$ terms is positive. Thus, we can apply the Hutchinson estimator to the corresponding matrix polynomial.

\begin{table}[htb!]\label{tb:1}
\centering
\begin{tabular}{cccc}
\hline\hline
$p$ & $\kappa$ & Running Time & Citation\\
\hline
$> 2$& no &  ${p \over \epsilon^3}\left(nnz(A)n^{1\over p + 1} + n^{1+{2\over 1 +p}}\right)\log {1\over \delta}$ & \cite{musco_et_al:LIPIcs:2018:8339, DBLP:journals/corr/MuscoNSUW17}, Theorem 32\\
\hline
$\le 2$ & no &  ${1 \over p^3\epsilon^{\max\{3, 1/p\}}}\left(nnz(A)n^{1 \over 1 + p} + n^{1 + {2 \over 1 + p}}\right)\log {1\over \delta}$ &\cite{musco_et_al:LIPIcs:2018:8339, DBLP:journals/corr/MuscoNSUW17}, Theorem 33\\
\hline
%$\le 2$ & no &  & ${1 \over p^3\epsilon^{\max\{3, 1/p\}}}\left[nnz(A)n^{1/p -1/2 \over 1/p + 1/2}\sqrt{ n d_s(A) \over nnz(A)} + \sqrt{nnz(A)}n^{4/p + 1 \over 2/p +1}\right]$ &\cite{musco_et_al:LIPIcs:2018:8339, DBLP:journals/corr/MuscoNSUW17}, Th. 33\\
%&&&or ${1 \over p^3\epsilon^{\max\{3, 1/p\}}}\left[nnz(A)n^{1 \over 1 + p} + n^{1 + {2 \over 1 + p}}\right]$& \\
%\hline
$> 0$ & no & $\left({\gamma(p)\over \epsilon^{2+1/p}} n^{1/p} nnz(A) + n\log n \right)\log {1\over \delta}$& Theorem \ref{th:21}\\
\hline
$\ge 1$ & yes & ${\kappa\over \epsilon^2} \left(p \log \kappa + \log {1\over \epsilon}\right) nnz(A) \log {1\over \delta} $& \cite{doi:10.1137/16M1078148}, Theorem 4.5 \\
%$\ge 1$ & yes & $\sigma_i \in [\sigma_{min}, \sigma_{max}]$ & $O\left({\kappa\over \epsilon^2} \log \left({1\over \delta}\right) \left(p \log \kappa + \log {1\over \epsilon}\right)\right) nnz(A)$& \cite{doi:10.1137/16M1078148} \\
%& &  where $\sigma_{min}, \sigma_{max}>0$ are known & & \\
%\hline
%$> 0$ & yes & ${\kappa\over \epsilon^2} \left(p \log \kappa + \log {1\over \epsilon} + \gamma(p) \right)nnz(A)\log {1\over \delta}$ & Th. \ref{th:22}\\
%\hline
\hline
$> 0$ & yes & ${\kappa\over \epsilon^2} \left(p \log \kappa + \log {1\over\epsilon}\right)nnz(A)\log {1\over \delta}$ & Remark \ref{cor:23}\\
\hline

\end{tabular}
\caption{Recent Results for $(1\pm \epsilon)$ approximations of $\|A\|_p^p$, w.p. at least $1-\delta$.
\\We omit $O()$ in all bounds. $\gamma(p)$ is a constant given in $(\ref{eq:blabla})$ that depends only on $p$.}
\end{table}

\subsection{Roadmap}
In Section \ref{sec:notations} we introduce necessary notations. Section \ref{sec:hitch} provides a Hutchinson estimator for a special matrix polynomial that will be used in our algorithms. Section \ref{sec:one} describes an algorithm for approximating $\|A\|_p^p$ that does not require knowledge of $k$.
 Section \ref{sec:two} describes an algorithm for approximating $\|A\|_p^p$ that depends on $k$.
 %Section \ref{sec:four} explains how to implement our algorithms in the streaming model and
 Section \ref{sec:three} contains necessary technical claims.

\subsection{Notations}\label{sec:notations}
We use the following symbols: $[x]$ be the integer part of $x$,
\begin{equation}\label{eq:1}
{p \choose k} = {p(p -1)\dots (p -k+1)\over k!}, k\ge 2, \ \ \ {p \choose 1} = p.
\end{equation}
%$\|X\|_2$ be a spectral norm of the matrix $X$.
We use $\log x$ to denote the natural logarithm and dedicate lower case Latin letters for constants and real variables and upper case Latin letters for matrices. Consider the Taylor expansion
\begin{equation}\label{eq:bla}
(1-x)^p = 1 + \sum_{k=1}^\infty {p \choose k} (-1)^k x^k, \ \ \ -1<x<1.
\end{equation}
Denote
\begin{equation}\label{eq:6}
h(x) = \sum_{k=1}^\infty (-1)^{k-1} {p \choose k} x^k.
\end{equation}
It follows from $(\ref{eq:bla})$ that for $|x| < 1$:
\begin{equation}\label{eq:5}
1 - (1-x)^p = h(x).
\end{equation}
Denote
\begin{equation}\label{eq:7}
h_m(x) = \sum_{k=m+1}^\infty (-1)^{k-1} {p \choose k} x^k,
\end{equation}
and
\begin{equation}\label{eq:8}
\tilde{h}_m(x) = \sum_{k=1}^m (-1)^{k-1} {p \choose k} x^k.
\end{equation}
According to $(\ref{eq:5})$,
\begin{equation}\label{eq:9}
1 - (1-x)^p = h(x) = \tilde{h}_m(x) + h_m(x).
\end{equation}
Denote by $H_t$ be the Hutchinson estimator \cite{hutchinson} for the trace:
\begin{equation}\label{eq:9999}
H_t(X) = {1\over t} \sum_{i=1}^t g_i^T X g_i,
\end{equation}
where $g_1,\dots, g_t \in R^n$ are i.i.d. vectors whose entries are independent Rademacher variables.
Note that $H_t(I_n) = n$.
We will be using the following well-known results.

\begin{theorem}\label{th:avrontol} (Roosta-Khorasani and Ascher \cite{roosta2015improvedArxiv}, Theorem $1$)
Let $Y\in R^{n\times n}$ be an SPSD matrix. Let $t> {8\over \epsilon^2}\log {1\over \delta}$. Then, with probability at least $1-\delta$:
\begin{equation}\label{eq:sdfg}
\left|H_t(Y) - \mathbf{tr}\ Y\right| \le \epsilon\ \mathbf{tr}\  Y.
\end{equation}
\end{theorem}

%When $\epsilon$ and $\delta$ are clear from the context, we will write $H_t(X)$ instead of $H_t(X,\epsilon, \delta)$. We also will be using the following well known result.%
%\begin{remark}
%Let $H_t$ be the Hutchinson estimator \cite{hutchinson} for the trace:
%$$
%H_t(X) = {1\over t} \sum_{i=1}^t g_i^T X g_i,
%$$
%where $g_1,\dots, g_t \in R^n$ are i.i.d. vectors whose entries are independent Rademacher variables. Then the output of Algorithm \ref{alg:lp set sektch} is $H_t(\tilde{h}_m(X))$. Also, because $g_i^Tg_i = n$, we have $H_t(I_n) = n$.
%\end{remark}

\begin{lemma}\label{lm:5} (Boutsidis, Drineas, Kambadur, Kontopoulou, and Zouzias \cite{boutsidis2017randomized}, Lemma $3$)
Let $A$ be a symmetric positive semidefinite matrix. There exists an algorithm that runs in time $O((n+nnz(A))\log n \log {1\over \delta})$ and outputs $\alpha$ such that, with probability at least $1-\delta$:
$$
{1\over 6}\lambda_1(A)\le \alpha \le \lambda_1(A).
$$
\end{lemma}

\section{Hutchinson Estimator for $\tilde{h}_m(X)$}\label{sec:hitch}
%Let $H_t$ be the Hutchinson estimator \cite{hutchinson} for the trace:
%$$
%H_t(X) = {1\over t} \sum_{i=1}^t g_i^T X g_i,
%$$
%where $g_1,\dots, g_t \in R^n$ are i.i.d. vectors whose entries are independent Rademacher variables.
The following simple algorithm computes $H_t(\tilde{h}_m(X))$.
%for approximating $\mathbf{tr}\ \tilde{h}_m(X)$ for a SPSD matrix $X$, where $\tilde{h}_m$ is defined in $(\ref{eq:8})$. The algorithm works correctly when $\tilde{h}_m(X)$ is SPSD which is the case in our applications.

\begin{algorithm}
	\caption{\label{alg:lp set sektch}}
	\begin{algorithmic}[1]
		\State \textbf{Input:} Matrix $X \in R^{n\times n}, t\ge 1$,
        \Comment $\tilde{h}_m$ is defined in $(\ref{eq:8})$ \ \ \ \ \
        \State \textbf{Output:} $H_t(\tilde{h}_m(X))$.
		\State \textbf{Initialization:}
%        \State $t = {8\over \epsilon^2} \log {1\over \delta}$
		\State Let $g_1,g_2,\dots, g_t \in \{-1, +1\}^n$ be i.i.d. random vectors \\ \ \ \ \ \ \  \ \ \ whose entries are independent Rademacher variables.
        \For{$i = 1,2,\dots, t$}
            \State $v_1 = Xg_i$, \ \  $u_1 = g_i^T v_1$
            \State $a_1 = p$
            \State $S_i^1 = a_1u_1$
            \For{$k = 2,\dots, m$}
                \State $v_k = Xv_{k-1}$
                \State $u_k = g_i^Tv_{k-1}$
                \State $a_k = a_{k-1}{p-(k-1)\over k}$
                \State $S_i^k = S_i^{k-1} + (-1)^{k-1}a_ku_k$
            \EndFor
        \EndFor
        \State Return $y = {1\over t}\sum_{i = 1}^t S_i$.
	\end{algorithmic}
\end{algorithm}

%\begin{remark}
%In this paper we need to compute a trace of a specific polynomial, namely $tr P(B)$ where $P(x) = (1+\beta_m) - \tilde{h}_m(1  -  \alpha^{-1}x)$, where $tilde{h}_m(x)$
%is given by $(\ref{eq:8})$. Thus, we can apply Algorithm $1$ with $a_0 = (1+\beta_m)$ and $a_i = (-1)^{i-1} {p \choose i}$.
%Note that $a_i$ can be computed recursively from $a_{i-1}$ and thus there is no need to store $a_0,\dots, a_t$ explicitly, in contrast with the coefficients for Chebyshev polynomials.
%\end{remark}

\begin{theorem}\label{th:five}
Let $X\in R^{n\times n}$ be a matrix and let $t\ge 1$ be an integer.
Then the output $y$ of Algorithm \ref{alg:lp set sektch} is $H_t(\tilde{h}_m(X))$ and the running time of the algorithm is $O(t \cdot nnz(X))$
\end{theorem}
\begin{proof}
For $k=1,\dots,m$ denote $a_k = {p \choose k}$ and note that $a_k = a_{k-1}{p-(k-1)\over k}$ and $\tilde{h}_m(X) = \sum_{k=1}^m (-1)^{k-1}a_kX^k$, by $(\ref{eq:8})$.
Fix $i \in \{1,\dots,t\}$ and note that after the $k$-th iteration we have
$$
u_k = g_i^TX^kg_i, \ \ \ S_i^k = g_i^T\left(\sum_{j=1}^k (-1)^{j-1}a_j X^j\right)g_i.
$$
Thus,
$
S_i^{m} = g_i^T(\tilde{h}_m(X))g_i.
$
Denote $Y= \tilde{h}_m(X)$; thus the algorithm outputs the estimator $y = {1\over t}\sum_{i=1}^t g_i^TYg_i = H_t(\tilde{h}_m(X))$ and the theorem follows.
The bound on the running time follows from direct computations.
\end{proof}

%\begin{theorem}\label{th:five}
%Let $X$ be an SPSD matrix and let $0<\epsilon, \delta <1$.
%Then the output $y$ of Algorithm \ref{alg:lp set sektch} with probability at least $1-\delta$ satisfies:
%$$
%(1-\epsilon)\mathbf{tr}\  \tilde{h}_m(X)\le y \le (1+\epsilon)\mathbf{tr}\  \tilde{h}_m(X).
%$$
%The running time of the algorithm is $O({1\over \epsilon^2} m \cdot nnz(X)\log {1\over \delta})$
%\end{theorem}
%\begin{proof}
%For $k=1,\dots,m$ denote $a_k = {p \choose k}$ and note that $a_k = a_{k-1}{p-(k-1)\over k}$ and $\tilde{h}_m(X) = \sum_{k=1}^m (-1)^ka_kX^k$, by $(\ref{eq:8})$.
%Fix $i \in \{1,\dots,t\}$ and note that after the $k$-th iteration we have
%$$
%u_k = g_i^TX^kg_i, \ \ \ S_i^k = g_i^T\left(\sum_{j=1}^k (-1)^ja_j X^j\right)g_i.
%$$
%Thus,
%$
%S_i^{m} = g_i^T(\tilde{h}_m(X))g_i
%$
%Denote $Y= \tilde{h}_m(X)$; thus the algorithm outputs the estimator $y = {1\over t}\sum_{i=1}^t g_i^TYg_i$. Recall that, by our assumption, $Y$ is an SPSD matrix. It is well known that, for $t=\Theta\left({1\over \epsilon^2}\log {1\over \delta}\right)$ such estimator indeed provides a multiplicative approximation for the trace of $Y$.
%Specifically, the theorem follows from the results of Roosta-Khorasani and Ascher \cite{roosta2015improved,roosta2015improvedArxiv}, see Theorem \ref{th:avrontol}. The bound on the running time follows from direct computations.
%\end{proof}

\section{Algorithm for estimates without $\kappa$}\label{sec:one}

\begin{algorithm}
	\caption{\label{alg:withoutkappa}}
	\begin{algorithmic}[1]
		\State \textbf{Input:} SPSD Matrix $A \in R^{n\times n}$, $p\ge 1, \epsilon\in (0,1), \delta\in (0,1)$.
        \State \textbf{Output:} $y$ such that, w.p. at least $1-\delta$, $(1-\epsilon)\|A\|_p^p \le y \le (1+\epsilon)\|A\|_p^p$.
		\State \textbf{Initialization:}
        \State $t > {8\over \epsilon^2} \log {2\over \delta}, m > 7\left({3c(p)\over p} {n\over\epsilon}\right)^{1/p}, \beta_m=  {c(p)\over p} (m+1)^{-p}$
        \Comment $c(p)$ is given in $(\ref{eq:4.5})$
		\State Compute $\alpha$ such that  $\lambda_1(A) \le \alpha \le 6\lambda_1(A)$.
        \Comment Use the algorithm from Lemma \ref{lm:5}.
        \State Return $y = \alpha^p \left[(1+\beta_m)n - H_t(\tilde{h}_m(I_n  -  \alpha^{-1}{A}))\right]$.
        \Comment Use Algorithm \ref{alg:lp set sektch}. $\tilde{h}_m$ is a function given by $(\ref{eq:8})$.
	\end{algorithmic}
\end{algorithm}

\begin{theorem}\label{th:21}
Let $A\in R^{n\times n}$ be an SPSD matrix and let $y$ be the output of Algorithm \ref{alg:withoutkappa} on input $A,\epsilon,\delta$.
Then, with probability at least $1-\delta$, we have:
$$
\left|\|A\|_p^p - y\right| \le \epsilon\|A\|_p^p.
$$
If we use Algorithm \ref{alg:lp set sektch} to compute $H_t$ and Lemma \ref{lm:5} to compute $\alpha$ then the running time of Algorithm \ref{alg:withoutkappa} is
$$
O\left(\left[{\gamma(p)\over \epsilon^{2+1/p}} n^{1/p} nnz(A) + n\log n \right]\log {1\over \delta}\right),
$$
where
\begin{equation}\label{eq:blabla}
\gamma(p) = {c(p)\over p}^{1/p}(1+6^p)^2.
\end{equation}
\end{theorem}
%\begin{theorem}\label{th:21}
%Let $\lambda_1(A) \le \alpha \le 6\lambda_1(A)$ from Lemma \ref{lm:5}. Let $c(p)$ be a constant that depends only on $p$, given in $(\ref{eq:4.5})$. Then, with the number of samples $t > {50\over \epsilon^2} \log (2/\delta)$ and
%$$
%m> \left({21c(p)\over p} {n\over\epsilon}\right)^{1/p},
%$$
%we have
%$$
%\left|\mathbf{tr} A^p - \alpha^p \left[(1+\beta_m)n - H_t(\tilde{h}_m(I_n  -  \alpha^{-1}{A}))\right]\right| \le \epsilon\mathbf{tr} A^p,
%$$
%with a probability at least $1-\delta$. Here $H_t$ is the Hutchinson estimator from $(\ref{eq:9999})$, $\tilde{h}_m$ is a function given by $(\ref{eq:8})$, and $\beta_m=  {c(p)\over p} (m+1)^{-p}$. If we use Algorithm \ref{alg:lp set sektch} to compute $H_t$ then the total running time of the algorithm is
%$$
%O\left(\left[{\gamma(p)\over \epsilon^{2+1/p}} n^{1/p} nnz(A) + n\log n \right]\log {1\over \delta}\right),
%$$
%where
%\begin{equation}\label{eq:blabla}
%\gamma(p) = {c(p)\over p}^{1/p}(1+6^p)^2.
%\end{equation}
%\end{theorem}
\begin{proof}
We have for $x\in [0,1]$, $x^p = [1-(1-x)]^p = 1 - h(1-x).$
Let $B$ be a SPSD matrix with $\lambda_1(B)\le 1$.
Denote $C = I_n-B$.
Since $\lambda_1(B) \in [0,1]$ we have that $C$ is also a SPSD matrix and
\begin{equation}\label{eq:41}
B^p = I_n - h(C).
\end{equation}

According to Claim \ref{cl:11} the matrix $(1+\beta_m)I_n - \tilde{h}_m(C)$ is SPSD. So, we may apply estimator $H_t$ to this matrix. We also have
$$
B^p = I_n - h(C) =  I_n - \tilde{h}_m(C) - h_m(C) = [(\beta_m+1)I_n - \tilde{h}_m(C)] - [{h}_m(C) + \beta_mI_n].
$$
Because $H_t(I_n) = n$,
$$
\left|\mathbf{tr}\ B^p - \left[(1+\beta_m)n - H_t(\tilde{h}_m(C))\right]\right|
 = \left|\mathbf{tr}\ B^p - H_t\left[(1+\beta_m)I_n - \tilde{h}_m(C)\right]\right|
$$
$$
\le \left|\mathbf{tr} \left[(\beta_m+1)I_n-\tilde{h}_m(C)\right]-H_t\left[(\beta_m+1)I_n-\tilde{h}_m(C)\right]\right| + |\mathbf{tr} (\beta_mI_n)| + |\mathbf{tr} ({h}_m(C))| = \Delta_1 + \Delta_2 + \Delta_3.
$$
Below we will bound each $\Delta$ separately.
\subsection*{Bounding $\Delta_1$ and $\Delta_2$}\label{sec:twotwo}

According to Theorem \ref{th:avrontol} for $t> {8\over \epsilon^2} \log (1/\delta)$ we have with probability at least $1-\delta$:
$$
\Delta_1 \le \epsilon \mathbf{tr} \left[(\beta_m+1)I_n-\tilde{h}_m(C)\right] = \epsilon n\beta_m + \epsilon \mathbf{tr} \left[I_n-\tilde{h}_m(C)\right]
$$
$$
= \epsilon n\beta_m + \epsilon \mathbf{tr} \left[I_n-{h}(C)\right] + \epsilon \mathbf{tr} {h}_m(C) = \epsilon n\beta_m + \epsilon \mathbf{tr}\ B^p + \epsilon \mathbf{tr}\ {h}_m(C).
$$
Note that $\Delta_2 = n\beta_m$.

\subsection*{Bounding $\Delta_3$}\label{sec:twotwo}

Denoting $\lambda_i = \lambda_i(B)$ we have, using Claim \ref{claim:2}:
$$
|\mathbf{tr}\ {h}_m(C)| = |\sum_{i=1}^n h_m(1-\lambda_i)| \le {c(p)\over p} \sum_{i=1}^n(1-\lambda_i)^{m+1}(m+1)^{-p} \le {c(p)\over p} n(m+1)^{-p} =  n\beta_m.
$$

\subsection*{The final bound for $B$}\label{sec:twotwo}
So,
$$
\Delta_1 + \Delta_2 +\Delta_3 \le (2\epsilon + 1) n\beta_m + \epsilon \mathbf{tr}\ B^p \le 3 n\beta_m + \epsilon \mathbf{tr}\ B^p.
$$
Thus, if $3 n\beta_m = {3c(p)\over p} n(m+1)^{-p} < \epsilon$, i.e., $m> \left({3c(p)\over p} {n\over\epsilon}\right)^{1/p}$ then
$$
\left|\mathbf{tr}\ B^p - H_t\left[(1+\beta_m)I_n - \tilde{h}_m(C)\right]\right| \le \epsilon + \epsilon \mathbf{tr}\ B^p,
$$
with probability at least $1-\delta.$

\subsection*{Back to matrix $A$}\label{sec:twotwo}
Let $\alpha$ be the constant from Lemma \ref{lm:5} and $B = A/\alpha$. Then $\mathbf{tr}\ B^p = {\alpha^{-p}}\mathbf{tr}\ A^p $ and $C = I_n - B = I_n - A/\alpha$. Recall that $\|A\|_p^p = \mathbf{tr}\ A^p$ and $y = \alpha^p \left[(1+\beta_m)n - H_t(\tilde{h}_m(I_n  -  \alpha^{-1}{A}))\right]$. So, from the above:
$$
\left|\|A\|_p^p - y \right|= \left|\mathbf{tr}\ A^p - \alpha^p \left[(1+\beta_m)n - H_t(\tilde{h}_m(I_n  -  \alpha^{-1}{A}))\right]\right| \le \alpha^p\epsilon \mathbf{tr}\ B^p + \alpha^p \epsilon
$$
$$
\le \epsilon \mathbf{tr}\ A^p + \epsilon (6\lambda_1(A))^p \le \epsilon(1+6^p)\mathbf{tr}\ A^p = \epsilon(1+6^p)\|A\|_p^p,
$$
with probability at least $1-2\delta.$ We obtain the result by substituting $\epsilon$ with $\epsilon\over (1+6^p)$, $\delta$ with $\delta \over 2$, and noting that $(1+6^p)^{1/p} \le 7$.
The bound on the running time follows from Theorem \ref{th:five} and Lemma \ref{lm:5}.
\end{proof}

\section{Algorithm for estimates with $\kappa$}\label{sec:two}

\begin{algorithm}
	\caption{\label{alg:withkappa}}
	\begin{algorithmic}[1]
		\State \textbf{Input:} SPSD Matrix $A \in R^{n\times n}$, $\kappa>0, p\ge 1, \epsilon\in (0,1), \delta\in (0,1)$.
        \State \textbf{Output:} $y$ such that, w.p. at least $1-\delta$, $(1-\epsilon)\|A\|_p^p \le y \le (1+\epsilon)\|A\|_p^p$.
		\State \textbf{Initialization:}
        \State $t > {8\over \epsilon^2} \log {2\over \delta}, m > 6\kappa \left[p \log (6\kappa) + \log \left({3\over\epsilon}\right) + \log\left({c(p)\over p}\right)\right]$
        \Comment $c(p)$ is given in $(\ref{eq:4.5})$
		\State Compute $\alpha$ such that  $\lambda_1(A) \le \alpha \le 6\lambda_1(A)$.
        \Comment Use the algorithm from Lemma \ref{lm:5}.
        \State Return $y = \alpha^p \left[ n - H_t(\tilde{h}_m(I_n-{\alpha}^{-1}A))\right]$.
        \Comment Use Algorithm \ref{alg:lp set sektch}. $\tilde{h}_m$ is a function given by $(\ref{eq:8})$.
	\end{algorithmic}
\end{algorithm}

\begin{theorem}\label{th:22}
Let $A\in R^{n\times n}$ be an SPSD matrix such that ${\lambda_1(A)\over \lambda_n(A)} \le \kappa$ and let $y$ be the output of Algorithm \ref{alg:withkappa} on input $A, \kappa,\epsilon,\delta$.
Then, with probability at least $1-\delta$, we have:
$$
\left|\|A\|_p^p - y\right| \le \epsilon\|A\|_p^p.
$$
If we use Algorithm \ref{alg:lp set sektch} to compute $H_t$ and Lemma \ref{lm:5} to compute $\alpha$ then the running time of Algorithm \ref{alg:withkappa} is

$$
O\left(\left({1\over \epsilon^2} nnz(A)\kappa \left[p \log (6\kappa) + \log \left({1\over\epsilon}\right) + \log\left({c(p)\over p}\right)\right]+ (n+nnz(A))\log n\right)\log {1\over \delta} \right)
$$
\end{theorem}

\begin{remark}\label{cor:23}
In \cite{doi:10.1137/16M1078148} the authors assume that they are given $\lambda_1(A)$ and $\lambda_n(A)$. This assumption is stronger than our assumption that we only know $\kappa = {\lambda_1(A)\over \lambda_n(A)}$. Under the stronger assumption of \cite{doi:10.1137/16M1078148} we do not need to estimate $\lambda_1(A)$ using Lemma \ref{lm:5}. Thus we don't need to perform Step $5$ of Algorithm \ref{alg:withkappa}. Hence the running time becomes
$$
O\left({1\over \epsilon^2} nnz(A)\kappa \left[p \log (6\kappa) + \log \left({1\over\epsilon}\right) + \log\left({c(p)\over p}\right)\right]\log {1\over \delta} \right),
$$
which, for constant $p$, is the same as in \cite{doi:10.1137/16M1078148}.
 \end{remark}

%\begin{theorem}\label{th:22}
%Let $\lambda_1(A) \le \alpha \le 6\lambda_1(A)$ from Lemma \ref{lm:5}.
%%Let $c(p) < (2+p)^{2p}$ be a constant that depends only on $p$, given in $(\ref{eq:4.5})$.
%Then with the number of samples $t > {8\over \epsilon^2} \log (2/\delta)$ and
%\begin{equation}\label{eq:43}
%m > 6\kappa \left[p \log (6\kappa) + \log \left({1\over\epsilon}\right) + \log\left({c(p)\over p}\right)\right],
%\end{equation}
%we have
%$$
%\left|\mathbf{tr}\ A^p - \alpha^p \left[ n - H_t(\tilde{h}_m(1-{\alpha}^{-1}A))\right]\right|\le 3\epsilon \mathbf{tr} A^p,
%$$
%with a probability at least $1-\delta$. Here $H_t$ is an estimator for the trace from Theorem \ref{th:five} and the function $\tilde{h}_m$ is given by $(\ref{eq:8})$.
%The running time of the algorithm is
%$$
%O\left({1\over \epsilon^2} nnz(A)\kappa \left[p \log (6\kappa) + \log \left({1\over\epsilon}\right) + \log\left({c(p)\over p}\right)\right]\log {1\over \delta} \right)
%$$
%\end{theorem}
%
\begin{proof} (of Theorem \ref{th:22})
%We can compute, with high probability, a real $\alpha$ such that $\lambda_1(A) \le \alpha \le 6\lambda_1(A)$.
As in Theorem \ref{th:21}, $\alpha$ is the constant from Lemma \ref{lm:5} and $B = A/\alpha$. Denote $\lambda_i = \lambda_i(B) = \lambda_i(A)/\alpha$. We have $\lambda_i\le {\lambda_1(A)/\alpha} < 1$ and
$$
\lambda_i \ge {\lambda_n(A)\over \alpha} = {\lambda_n(A)\over \lambda_1(A)}{\lambda_1(A)\over \alpha} > {1\over 6\kappa},
$$
so $0<1-\lambda_i<1-{1\over 6\kappa}$ for $1\le i\le n$. From here and Claim \ref{cl:24}, with $a = 1-{1\over 6\kappa}$,
$$
1-\tilde{h}_m(1-\lambda_i)> 0, 1\le i\le n,
$$
where
\begin{equation}\label{eq:33}
m> 6\kappa (\log(c(p)/p) + p \log(6\kappa)).
\end{equation}
From here we conclude that for such $m$ the matrix $I_n - \tilde{h}_m(I-B) =I_n - \tilde{h}_m(C)$ is SPSD.
%Hence we may apply the trace estimator $H_t$ to this matrix.
In addition,
$$
B^p = \left[I_n - (I_n - B)\right]^p = I_n - h(C) = [I_n - \tilde{h}_m(C)] - {h}_m(C).
$$
Since $H_t(I_n) = n$,
\begin{equation}\label{eq:378}
\begin{split}
\left|\mathbf{tr}\ B^p - n + H_t(\tilde{h}_m(C))\right| = \left|\mathbf{tr}\ B^p - H_t(I_n - \tilde{h}_m(C))\right| \\
 \le |\mathbf{tr} [I_n - \tilde{h}_m(C)] - H_t(I_n - \tilde{h}_m(C))| & + |\mathbf{tr} \ {h}_m(C)| = \Delta_1 + \Delta_2.
\end{split}
\end{equation}
To conclude, we will provide estimates for $\Delta_1$ and $\Delta_2$.
\subsection*{Estimate for $\Delta_1$}
According to Theorem \ref{th:avrontol}, with probability at least $1-\delta$,
$$
\Delta_1 \le \epsilon \mathbf{tr}\left[I_n - \tilde{h}_m(C)\right] = \epsilon \mathbf{tr}\left[I_n - {h}(C)\right]  + \epsilon \mathbf{tr}\ {h}_m(C).
$$
But $I_n - {h}(C) = B^p$ and so
\begin{equation}\label{eq:34}
\Delta_1 \le  \epsilon \mathbf{tr}\ B^p  + \epsilon \Delta_2.
\end{equation}

\subsection*{Estimate for $\Delta_2$}
We have $1/(6\kappa) \le \lambda_i \le 1$. Applying Claim \ref{cl:25} with $a={1\over 6\kappa}$ we get:
$$
|h_m(1-\lambda_i)|<\epsilon \lambda_i^p,
$$
if $m$ satisfies $(\ref{eq:29})$. For such $m$:
$$
\Delta_2 = |\mathbf{tr} \ {h}_m(C)| \le \sum_{i=1}^n  |h_m(1-\lambda_i)| \le \epsilon \sum_{i=1}^n \lambda_i^p = \epsilon \mathbf{tr}\ B^p.
$$
\subsection*{The final estimate for matrix $B$}
If
\begin{equation}\label{eq:43}
m \ge 6\kappa \left[p \log (6\kappa) + \log \left({1\over\epsilon}\right) + \log\left({c(p)\over p}\right)\right].
\end{equation}
then  $m$ satisfies $(\ref{eq:33})$, so for such $m$ from $(\ref{eq:378})$ and $(\ref{eq:34})$:
\begin{equation}\label{eq:36}
\left|\mathbf{tr}\ B^p - \left[n - H_t(\tilde{h}_m(C))\right]\right| \le \Delta_1 + \Delta_2 \le \epsilon \mathbf{tr}\ B^p + (1+\epsilon)\Delta_2 \le \epsilon (2+\epsilon) \mathbf{tr}\ B^p \le 3\epsilon \mathbf{tr}\ B^p.
\end{equation}
\subsection*{Back to matrix $A$}
We have $B = A/\alpha$ and $\mathbf{tr}\ B^p = {\alpha^{-p}}\mathbf{tr}\ A^p $.
Recall that $\|A\|_p^p = \mathbf{tr}\ A^p$ and $y = \alpha^p \left[ n - H_t(\tilde{h}_m(1-{\alpha}^{-1}A))\right]$.
So we get from $(\ref{eq:36})$, w.p. at least $1-2\delta$:
$$
\left|\|A\|_p^p - y \right| = \left|\mathbf{tr}\ A^p - \alpha^p \left[n - H_t(\tilde{h}_m(C))\right]\right| \le  \alpha^p (3\epsilon) \mathbf{tr}\ B^p \le 3\epsilon\ \mathbf{tr}\ A^p = 3\epsilon\|A\|_p^p,
$$
for $m$ satisfying $(\ref{eq:43})$. Since $C = I_n-\alpha^{-1}A$ and by substituting $\epsilon$ with $\epsilon\over 3$ and $\delta$ with $\delta\over 2$ the theorem follows.
\end{proof}

\section{Technical Lemmas}\label{sec:three}

%--------------- Claim

\begin{claim}\label{claim:1}
There exists a constant $c(p)$, (defined in $(\ref{eq:4.5})$), that depends only on $p$, such that
$$
\left|{p \choose k}\right| \le c(p)(k+1)^{-(p + 1)},
$$
for any $k> [p]+1$.
\end{claim}
\begin{proof}
Denote
$$
c_1(p)=
\begin{cases}
\prod_{i=1}^{[p]}\left|p - i + 1 \over i\right|,\ \ \ p >1, \\
1,\ \ \ 0 \le p \le 1.
\end{cases}
$$
If $i>[p]+1$ then $\left|p - i + 1 \over i\right| = 1 - {p +1 \over i}$. Thus, for $k>[p]+1$:
\begin{equation}\label{eq:2}
\left|{p \choose k}\right| = c_1(p) {p - [p]\over [p] + 1}\prod_{i=[p]+2}^k\left(1 - {p +1 \over i}\right).
\end{equation}
Since $0<1-x<e^{-x}$ for $0<x<1$ and ${p - [p]\over [p] + 1}<1$, we have:
\begin{equation}\label{eq:3}
\left|{p \choose k}\right| \le c_1(p) \exp\left(-(1+p)\sum_{i=[p]+2}^k {1 \over i}\right).
\end{equation}

Further, it is known from Langrange formula that ${1\over i} \ge \ln (i+1) - \ln i$ and thus
\begin{equation}\label{eq:4}
\sum_{i=[p]+2}^k {1 \over i} \ge \sum_{i=[p]+2}^k (\ln (i+1) - \ln i) = \ln (k+1) - \ln ([p] + 2).
\end{equation}
From here and $(\ref{eq:3})$ we obtain
\begin{equation}\label{eq:4}
\left|{p \choose k}\right| \le c_1(p) \exp\left(-(1+p)(\ln (k+1) - \ln ([p] + 2))\right) = c_1(p) (2+[p])^{p+1}(k+1)^{-(p+1)}.
\end{equation}
Denote
\begin{equation}\label{eq:4.5}
c(p) = c_1(p)(2+[p])^{p+1}.
\end{equation}
The claim follows.
\end{proof}

%--------------- Claim

\begin{claim}\label{claim:1.5}
For all $p>0$,
$$
c(p) > p.
$$
%where $c(p)$ is from Claim $\ref{claim:1}$.
\end{claim}
\begin{proof}
Assume first that $p \in (0,1).$ Then $[p] = 0$ and $c_1(p) = 1$ so $c(p) = 2^{p+1}> p$.
Let now $p \in [1,2)$. Then $[p] = 1, c_1(p) = p$ and $c(p) = p 3^{p+1} > p$.
Assume now $p \ge 2$. Then $[p] \ge 2$ and
\begin{equation}\label{eq:1.6}
c(p) = p \prod_{i=2}^{[p]}\left|p - i + 1 \over i\right|(2+[p])^{p+1}.
\end{equation}
For $i\le [p]$ we have
$$
{p - i + 1 \over i} = {p  + 1 \over i} - 1 \ge {p  + 1 \over [p]} - 1 \ge {p  + 1 \over p} - 1 = {1\over p}.
$$
So,
$$
\prod_{i=2}^{[p]}\left|p - i + 1 \over i\right|(2+[p])^{p+1} \ge \left( {1\over p}\right)^{p-1}(2+[p])^{p+1} =  \left( {2+[p]\over p}\right)^{p-1}(2+[p])^{2}>1,
$$
where the last inequality holds because $2+[p] > p$. From here and $(\ref{eq:1.6})$ the claim follows.
\end{proof}

%--------------- Claim

\begin{claim}\label{claim:2}
For $m > [p]+1$ and $|x| \le 1$:
$$
\left|h_m(x)\right|\le {c(p)\over p} |x|^{m+1}(m+1)^{-p}.
$$
%where $c(p)$ is from $(\ref{eq:4.5})$.
\end{claim}
\begin{proof}
We have:
$$
|h_m(x)| \le \sum_{k=m+1}^\infty \left|{p \choose k}\right| |x|^k \le |x|^{m+1}\sum_{k=m+1}^\infty \left|{p \choose k}\right|  \le |x|^{m+1}c(p)\sum_{k=m+1}^\infty (k+1)^{-(p+1)},
$$
where the first inequality follows from the definition of $h_m(x)$ in $(\ref{eq:7})$, the second inequality follows since $|x| \le 1$ and the third inequality follows from Claim \ref{claim:1}.Let us estimate the last sum\footnote{This and some other bounds in this paper are well known and we include their proofs for completeness.}. For any $k\le x\le k+1$ we have $(k+1)^{-(p+1)}\le x^{-(p+1)}$. Integrating over $[k, k+1]$ we obtain
$$
(k+1)^{-(p+1)} \le \int_{k}^{k+1}x^{-(p+1)} dx.
$$
Thus
$$
\sum_{k=m+1}^\infty (k+1)^{-(p+1)} \le \sum_{k=m+1}^\infty \int_{k}^{k+1}x^{-(p+1)} dx = \int_{m+1}^{\infty}x^{-(p+1)} dx = {1\over p}(m+1)^{-p}.
$$
The claim follows from here and the above bound.
\end{proof}

%--------------- Claim

\begin{claim}\label{cl:11}
For any $x\in [0,1]$ and $m > [p] +1$:
$$
1 + {c(p)\over p} (m+1)^{-p} - \tilde{h}_m(x) \ge (1-x)^p \ge 0.
$$
\end{claim}
\begin{proof}
We have $1- \tilde{h}_m(x) = [1-h(x)] + {h}_m(x) = (1-x)^p +{h}_m(x)$, and, by Claim \ref{claim:2}, $|h_m(x)|\le {c(p)\over p} (m+1)^{-p}
$.
So $1- \tilde{h}_m(x) \ge (1-x)^p - {c(p)\over p}(m+1)^{-p}$.
The claim follows.
\end{proof}

\begin{claim}\label{cl:24}
Let $0<a<1$. If
$$
m> {\log (c(p)/p) - p \log (1-a)\over 1-a} - 1
$$
then $1-\tilde{h}_m(x)>0$ for $0<x\le a$.
\end{claim}
\begin{proof}
%We have
%$$
%1- \tilde{h}_m(x) = [1- {h}(x)]+ {h}_m(x) = (1-x)^p + {h}_m(x)
%$$
%(because $1- {h}(x) = (1-x)^p$). For $0<x\le a$ we have $(1-x)^p \ge (1-a)^p$ and
Claim \ref{claim:2} yields
$$
1- \tilde{h}_m(x) = (1-x)^p + h_m(x) \ge (1-a)^p - |h_m(x)|\ge (1-a)^p - {c(p)\over p} x^{m+1}(m+1)^{-p} \ge
%(1-a)^p - {c(p)\over p} a^{m+1}(m+1)^{-p-1} \ge
(1-a)^p - {c(p)\over p} a^{m+1}.
$$
Solving the inequality $(1-a)^p - {c(p)\over p} a^{m+1} > 0$ with respect to $m$ we get
\begin{equation}\label{eq:jdsfjhsdf}
m> {\log (c(p)/p) - p \log (1-a)\over -\log a} - 1.
\end{equation}
Now we use the inequality
\begin{equation}\label{eq:65}
{-1\over \log x}\le {1\over 1-x},\ \ \ \ 0<x<1.
\end{equation}
To check it, note that this inequality is equivalent to
$$
g(x) \defeq  {1-x + \log x} < 0, \ \ 0<x<1.
$$
Since $g'(x) = -1 +{1\over x} > 0,$
for $0<x<1$ then $g(x) < g(1) = 0.$

\noindent
We have, according to Claim \ref{claim:1.5}, $\log (c(p)/p) > 0.$ Hence $(\ref{eq:65})$ yields
$$
{\log (c(p)/p) - p \log (1-a)\over -\log a} < {\log (c(p)/p) - p \log (1-a)\over 1-a}.
$$
From here and $(\ref{eq:jdsfjhsdf})$ the claim follows.
\end{proof}

\begin{claim}\label{cl:25}
Let $0<a<1$ and
\begin{equation}\label{eq:29}
m> {1\over a}\left[p \log \left({1\over a}\right) + \log \left({1\over \epsilon}\right) + \log \left({c(p)\over p}\right) \right] - 1.
\end{equation}
Then $|{h}_m(1-x)| < \epsilon x^p$ for $a\le x\le 1$.
\end{claim}
\begin{proof}
According to Claim \ref{claim:2},
$$
x^{-p} |{h}_m(1-x)| \le x^{-p} {c(p)\over p}(1-x)^{m+1}(m+1)^{-p} \le {c(p)\over p}(1-a)^{m+1}a^{-p},
$$
because $a\le x\le 1$. Now we choose $m$ such that last expression will be less than $\epsilon$. Solving the corresponding inequality, we get
\begin{equation}\label{eq:28}
m > {p \log \left({1\over a}\right) + \log \left({1\over \epsilon}\right) + \log \left({c(p)\over p}\right) \over -\log (1-a)} - 1.
\end{equation}
But we know from above that
$$
{1\over -\log(1-a)}\le {1\over a}.
$$
We also know from Claim \ref{claim:1.5} that $\log (c(p)/p) > 0$ so the right hand side of $(\ref{eq:28})$ is less than the right hand side of $(\ref{eq:29})$. Hence the claim follows.

\end{proof}

\bibliographystyle{abbrv}
\bibliography{bib}

\end{document}